\documentclass[11pt]{article}

\usepackage{fullpage}

\usepackage{nicefrac}
\usepackage{xcolor}
\usepackage{amsmath,amssymb,amsthm,mathtools}
\usepackage{enumitem}
\usepackage{microtype}
\usepackage{soul}
\usepackage{thm-restate}
\usepackage[ruled,vlined,linesnumbered]{algorithm2e}

\usepackage[hidelinks]{hyperref}
\usepackage[capitalize]{cleveref}
\usepackage{tikz,adjustbox}
\usetikzlibrary{arrows.meta,calc,fit,backgrounds}
\usepackage{amsthm}

\newtheorem{theorem}{Theorem}
\newtheorem{lemma}{Lemma}
\newtheorem{corollary}{Corollary}
\newtheorem{contract}{Contract}
\newtheorem{remark}{Remark}
\newtheorem{problem}{Problem}
\newtheorem{definition}{Definition}

\crefname{problem}{Problem}{Problems}
\crefname{contract}{Contract}{Contracts}
\crefname{line}{Line}{Lines}
\Crefname{line}{Line}{Lines}
\makeatletter
\renewcommand{\nllabel}[1]{%
  \begingroup
  \def\@currentcounter{AlgoLine}%
  \label{#1}%
  \endgroup
}
\makeatother

\newcommand{\norm}[1]{\left\lVert #1\right\rVert_\infty}
\newcommand{\clip}{\operatorname{clip}}

\newcommand{\Solve}{\mathsf{Solve}}
\newcommand{\Label}{\mathsf{Label}}
\newcommand{\Amplify}{\mathsf{Amplify}}
\newcommand{\ThinLabel}{\mathsf{ThinLabel}}
\newcommand{\Oracle}{\mathsf{Oracle}}

\newcommand{\eps}{\varepsilon}

\renewcommand{\L}{\Phi}
\def\[#1\]{\begin{align*}#1\end{align*}}

\NewDocumentCommand{\numberthis}{om}{%
  \IfNoValueTF{#1}{%
    \refstepcounter{equation}\tag{\theequation}%
  }{%
    \tag{#1}%
  }%
  \label{#2}%
}

\newcommand{\repeatthmname}{}
\newtheorem*{repeatthminner}{\protect\repeatthmname}
\newenvironment{repeatthm}[1]{%
  \renewcommand{\repeatthmname}{Theorem~\ref{#1}}%
  \begin{repeatthminner}%
}{%
  \end{repeatthminner}%
}

\usepackage{xcolor}
\definecolor{myOrange}{RGB}{235, 105, 92}
\definecolor{myMaroon}{RGB}{128, 0, 0}
\definecolor{maroon}{rgb}{0.5,0.0,0.0}
\definecolor{heavymaroon}{rgb}{0.3,0.0,0.0}

\definecolor{mygreen}{RGB}{160, 200, 140}

\usepackage{hyperref}
\hypersetup{ 
colorlinks=true, 
linkcolor=myOrange, 
citecolor=myMaroon, filecolor=cyan, 
urlcolor=mygreen
}

\setlist{nosep,leftmargin=*}
\title{A Quasipolynomial Algorithm for $\ell_\infty$ Contractions}
\author{
    Ashkan Soleymani\\
    MIT\\
    \texttt{ashkanso@mit.edu}\\
    \and
    Gabriele Farina\\
    MIT\\
    \texttt{gfarina@mit.edu}\\
}
\date{}

\begin{document}
\maketitle

\begin{abstract}
    We give an algorithm for computing an $\eps$-approximate fixed point of any $d$-dimensional nonexpansion with respect to the $\ell_\infty$ norm in time $\log(1/\eps)^{O(\log^2 d)}$. At its core, our method uses a boosting argument to geometrically shrink a search region. Among other consequences, our algorithm implies that simple stochastic games can be solved in quasipolynomial time.
\end{abstract}

\section{Introduction}
\label{sec:introduction}

Let $f:[0,1]^d\to[0,1]^d$ be an $\ell_\infty$-contraction. Thus, for some $\gamma\in(0,1]$,
\[
    \norm{f(x)-f(y)}
    \leq (1-\gamma)\norm{x-y}
    \qquad\text{for all }x,y\in[0,1]^d.
\]
By Banach's fixed-point theorem, $f$ has a unique fixed point
$x^*\in[0,1]^d$, satisfying $f(x^*)=x^*$.
We distinguish two notions of approximation.
For $\eps>0$, a point $x\in[0,1]^d$ is an
\emph{$\eps$-approximate fixed point} if
$\norm{f(x)-x}\leq\eps$, and an
\emph{$\eps$-near fixed point} if
$\norm{x-x^*}\leq\eps$. Contractivity ensures that every $\eps$-near fixed point is
a $2\eps$-approximate fixed point, while every
$\eps$-approximate fixed point lies within distance
$\eps/\gamma$ of the unique fixed point $x^*$.

The classical Picard iteration~\cite{banach1922sur} starts from an
arbitrary point $x^{(0)}\in[0,1]^d$ and repeatedly applies $f$,
setting $x^{(t+1)}=f(x^{(t)})$.
Although the iterates converge to $x^*$, the standard convergence
bound gives $O(\gamma^{-1}\log(1/\eps))$ evaluations to obtain an
$\eps$-approximate fixed point.
This bound can be exponential in the input size when the contraction
gap $\gamma$ is exponentially small, as can occur in Shapley's
discounted stochastic games~\cite{shapley1953stochastic}.

A central open problem at the intersection of optimization and computer science is whether an $\eps$-approximate fixed point
can be computed in time polynomial in $d$, $\log(1/\eps)$, and $\log(1/\gamma)$.
A polynomial-time solution would settle the long-standing complexity
of \emph{Condon's simple stochastic games} and, through known reductions,
yield polynomial-time algorithms for stochastic parity and
mean-payoff games~\cite{condon1992complexity,andersson2009complexity}.
It would also allow polynomial-time approximation of the values
of \emph{Shapley's stochastic games}, in which players choose their
actions simultaneously~\cite{shapley1953stochastic}.
Shapley's model further laid the foundations for Markov games,
a standard framework for multi-agent reinforcement
learning~\cite{littman1994markov}. More broadly, fixed-point problems play an important role in
optimization and distributed computation~\cite{bertsekas1983distributed,bertsekas1989parallel},
discounted dynamic programming and reinforcement
learning~\cite{blackwell1965discounted,puterman1994markov,sutton2018reinforcement},
differential and integral equations~\cite{banach1922sur,coddington1955theory},
and formal verification, particularly modal $\mu$-calculus
model checking~\cite{kozen1983results,emerson2001model}.

Along this avenue of investigation, Chen, Li, and Yannakakis showed that an $\eps$-approximate fixed
point of a nonexpansion can be found with $O(d^2\log(1/\eps))$ queries~\cite{chen2024computing},
later improved to $O(d\log(1/\eps))$ in the journal
version~\cite{chen2025computing}.
Despite their polynomial query complexity, these algorithms
incur exponential computational overhead in selecting the query points.
Feodorov and Haslebacher~\cite{feodorov2026faster}
subsequently bounded both query complexity and total running time by
$
    \log(1/\eps)^{O(\sqrt d\log d)}.
$
An account of further related work is deferred to \cref{sec:related-work}.

\paragraph{Result and consequences.}
We give the first quasipolynomial-time algorithm for computing
near fixed points of $\ell_\infty$-contractions.
In particular, we present a deterministic algorithm
$\Solve(f,\eps)$ with the following property.

\begin{theorem}\label{thm:main}
Fix $d\ge2$, a contraction gap $\gamma\in(0,1]$, and a desired nearness tolerance $0<\eps\le1/4$. For a $(1-\gamma)$-contraction $f:[0,1]^d\to[0,1]^d$, the total number of arithmetic operations and exact evaluations of $f$ performed by $\Solve(f,\eps)$ to return an $\eps$-\emph{near} fixed point of $f$ is upper bounded by
\[
    \numberthis{eq:solve-complexity}
\left(\log\frac1{\gamma\eps}\right)^{C \log^2 d},
\]
for a universal constant $C$.
\end{theorem}

Our main theorem also has direct and well-understood implications for computing an $\eps$-\emph{approximate} fixed point of \emph{nonexpansions} (i.e., $1$-Lipschitz functions $f$) \cite{chen2024computing}. Specifically, given a nonexpansion $f$ and a desired approximation $0<\eps<1$, the map $g\coloneqq(1-\eps/2)f$ is a
contraction with gap $\eps/2$, and every
$(\eps/4)$-near fixed point of $g$ is an
$\eps$-approximate fixed point of $f$. 
\begin{corollary}
    For any $d\ge2$, $0<\eps\le1/2$, and
    $\ell_\infty$-nonexpansive map
    $f:[0,1]^d\to[0,1]^d$, there exists a deterministic algorithm for computing an $\eps$-approximate fixed point of $f$ with complexity
    $\log(1/\eps)^{C \log^2 d}$, for a universal constant $C$.
\end{corollary}

Compared with the previous bound of Feodorov and
Haslebacher~\cite{feodorov2026faster}, our result reduces
the exponent from $O(\sqrt d\log d)$ to $O(\log^2 d)$,
bringing the total running time from subexponential
to quasipolynomial.
In particular, when $\log(1/\eps)=d^{O(1)}$, the
running-time bound improves from
$\exp(O(\sqrt d\log^2 d))$ to $\exp(O(\log^3 d))$.

\begin{figure}[!t]
\centering
\begin{adjustbox}{max width=\linewidth}
\definecolor{AccentMaroon}{HTML}{7A1F2B}
\definecolor{AccentFill}{HTML}{EADADA}
\definecolor{SoftGray}{HTML}{F0F0F0}
\definecolor{MutedText}{HTML}{555555}
\begin{tikzpicture}[
    x=1cm,y=1cm,
    font=\footnotesize,
    base/.style={
      draw=black, fill=white, line width=.73pt,
      rounded corners=2.4pt, align=center,
      inner xsep=5pt, inner ysep=3pt
    },
    gray/.style={base,fill=SoftGray},
    accent/.style={base, draw=AccentMaroon, fill=AccentFill, line width=1pt},
    spine/.style={
      text width=4.55cm, minimum height=1.05cm,
      inner xsep=5pt, inner ysep=3pt
    },
    side/.style={text width=3.05cm, minimum height=.96cm},
    decision/.style={
      text width=5.0cm, minimum height=2.05cm,
      inner xsep=9pt, inner ysep=7pt, align=left
    },
    group/.style={draw=black!60, dashed, rounded corners=3pt, line width=.7pt},
    to/.style={
      -{Latex[length=2.35mm,width=1.6mm]},
      draw=black, line width=.85pt,
      shorten <=1.0pt, shorten >=1.0pt
    },
    both/.style={
      {Latex[length=2.15mm,width=1.50mm]}-{Latex[length=2.15mm,width=1.50mm]},
      draw=black, line width=1.00pt,
      shorten <=0.8pt, shorten >=0.8pt
    },
    bothCompact/.style={
      {Latex[length=1.65mm,width=1.50mm]}-{Latex[length=1.65mm,width=1.50mm]},
      draw=black, line width=1.00pt,
      shorten <=0.6pt, shorten >=0.6pt
    }
]
\node[accent,spine] (fixed) at (0,13.54)
  {Approximate Fixed Points of\\$\ell_\infty$-Nonexpansive Maps};
\node[gray,spine] (monotone) at (0,11.98)
  {Approximate Fixed Points of\\Monotone $\ell_\infty$-Contractions};
\node[gray,spine] (shapley) at (0,10.42)
  {Discounted Shapley\\Stochastic Games};
\node[gray,spine] (discounted) at (0,8.86)
  {Turn-Based Discounted\\Stochastic Games};
\node[accent,spine] (ssg) at (0,7.30)
  {Simple Stochastic Games\\(SSGs)};

\draw[to] (monotone.north) -- (fixed.south);
\draw[to] (shapley.north) -- (monotone.south);
\draw[to] (discounted.north) -- (shapley.south);
\draw[bothCompact] (ssg.north) -- (discounted.south);

\node[base,side] (stochparity) at (-4.92,7.30)
  {Stochastic\\Parity Games};
\node[base,side] (stochmean) at (4.92,7.30)
  {Stochastic\\Mean-Payoff Games};
\draw[both] (stochparity.east) -- (ssg.west);
\draw[both] (ssg.east) -- (stochmean.west);

\begin{scope}[on background layer]
\node[group,fit=(discounted)(ssg)(stochparity)(stochmean),
      inner xsep=11pt,inner ysep=6pt] (core) {};
\end{scope}
\node[anchor=west,font=\scriptsize\itshape,text=MutedText,
      fill=white,inner xsep=3pt,inner ysep=0pt]
  at ($(core.north west)+(0.20,-0.01)$)
  {Polynomial-Time Turing Equivalence};

\node[base,side] (parity) at (-4.92,5.83) {Parity Games};
\node[base,side] (mean) at (4.92,5.83)
  {Deterministic\\Mean-Payoff Games};
\draw[to] (parity.east) -- (mean.west);
\draw[to] (mean.north) -- (stochmean.south);

\node[base,side] (mu) at (-4.92,4.11)
  {Modal $\mu$-Calculus\\Model Checking};
\draw[both] (mu.north) -- (parity.south);

\node[base,decision] (eqdet) at (4.92,3.25) {
  \textbf{Equivalent Decision Problems}\\[4pt]
  $\bullet$\enspace One-Dimensional Energy Games\\[2pt]
  $\bullet$\enspace Tropical Linear Feasibility\\[2pt]
  $\bullet$\enspace Tropical Linear Regression\\[2pt]
  $\bullet$\enspace Max-Atom Satisfiability
};
\draw[both] (eqdet.north) -- (mean.south);

\node[font=\scriptsize,text=MutedText,align=left,anchor=west]
  at (-6.7,2.1)
  {$\mathbf{A\longrightarrow B:\ A\leq_T^{p}B}$
   \qquad
   $\mathbf{A\longleftrightarrow B:\ A\equiv_T^{p}B}$};
\end{tikzpicture}
\end{adjustbox}

\caption{Polynomial-time reduction hierarchy connecting approximate
$\ell_\infty$ fixed-point computation to stochastic games and related
decision problems.
The dashed box groups turn-based, perfect-information stochastic
games that are polynomial-time Turing equivalent~\cite{andersson2009complexity}.
The upper reductions arise from the fixed-point formulation of
discounted stochastic games, where sufficiently accurate approximations
also yield exact solutions to SSGs~\cite{shapley1953stochastic,batziou2025monotone}. The lower reductions connect parity games, modal
$\mu$-calculus model checking, and decision problems equivalent
to deterministic mean-payoff games~\cite{jurdzinski1998deciding,
emerson2001model,brim2011faster,akian2012tropical,
akian2023tropical,atserias2009mean}. The red boxes highlight our main result and its application to SSGs,
the gray boxes indicate intermediate problems, and the white boxes
show further consequences.}
\label{fig:reduction-hierarchy-uniform}
\end{figure}
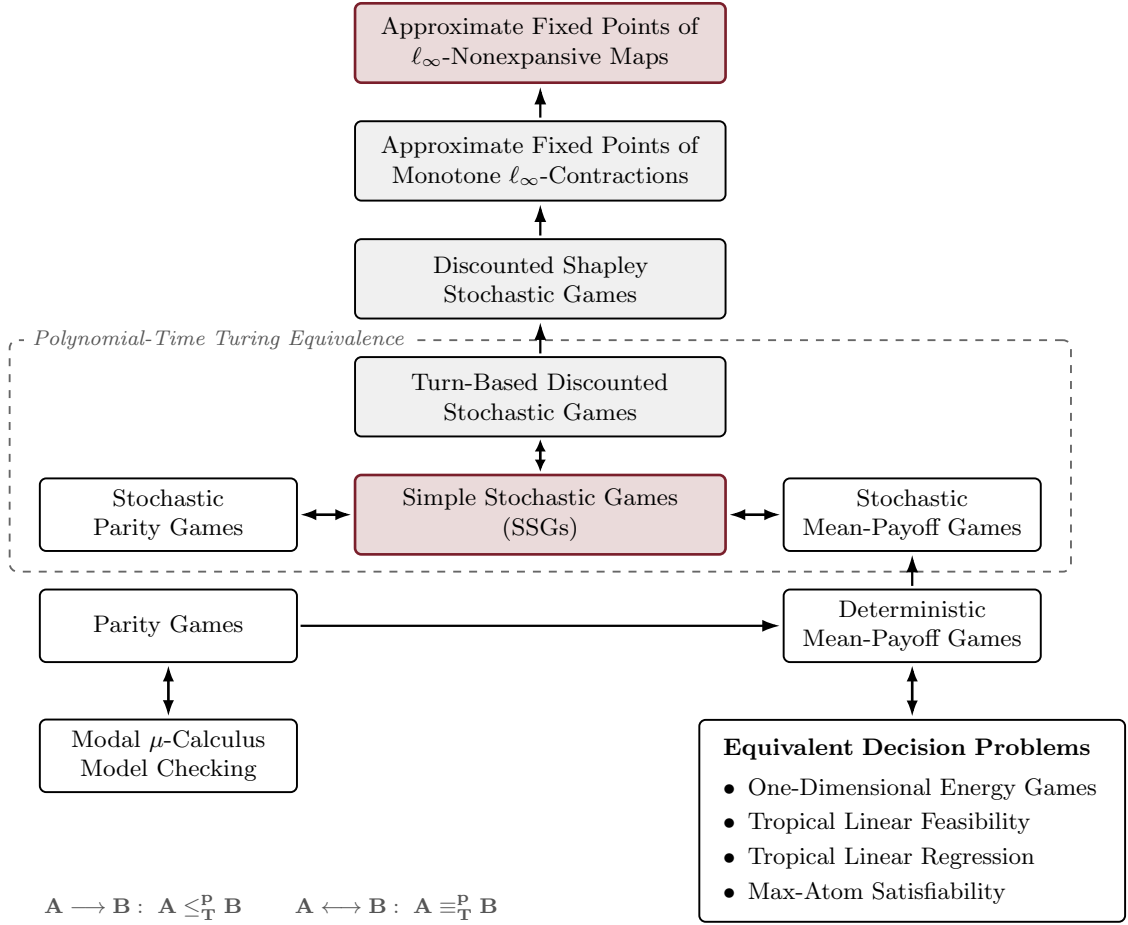

\paragraph{Applications.}
Our result yields quasipolynomial-time algorithms for several
stochastic games and related decision problems, as summarized in
\Cref{fig:reduction-hierarchy-uniform}.
In particular, Condon's simple stochastic games (SSGs) reduce
to computing sufficiently accurate fixed points of discounted
value operators~\cite{condon1992complexity,feodorov2026faster}.
The reduction uses polynomial-bit parameters and efficiently
evaluable operators, so our algorithm solves SSGs exactly in
quasipolynomial time.
By the polynomial-time Turing equivalences of Andersson and
Miltersen~\cite{andersson2009complexity}, the same bound extends
to turn-based stochastic parity, mean-payoff, discounted-payoff,
and terminal-payoff games.
For Shapley's general discounted stochastic games, which allow
simultaneous actions, our algorithm also yields quasipolynomial-time
approximation of their value vectors~\cite{shapley1953stochastic}.

The reductions in \Cref{fig:reduction-hierarchy-uniform}
extend these consequences to deterministic mean-payoff games,
parity games, and modal $\mu$-calculus model
checking~\cite{jurdzinski1998deciding,emerson2001model}.
The corresponding decision problems for one-dimensional energy
games, tropical linear feasibility, tropical linear regression,
and max-atom satisfiability are polynomial-time equivalent to
deterministic mean-payoff games~\cite{brim2011faster,
akian2012tropical,akian2023tropical,atserias2009mean}.
Although quasipolynomial-time algorithms are already known
for parity games~\cite{calude2017deciding}, our approach
provides a common route to these problems through
nonexpansive fixed-point computation.

\section{Further Related Work}
\label{sec:related-work}

Concurrent with our work, OpenAI released a broad collection of
preprints spanning several areas of mathematics, including
quasipolynomial-time algorithms for deterministic mean-payoff games,
stochastic mean-payoff games, and mean-payoff parity
games~\cite{openai2026meanpayoff,openai2026stochasticmeanpayoff,
openai2026meanpayoffparity}.
A randomized algorithm for deterministic mean-payoff games was also
included~\cite{openai2026randomizedmeanpayoff}.
These algorithms exploit the structure of the underlying games,
in particular the monotonicity of their value operators. Our algorithm, in contrast, applies to arbitrary
$\ell_\infty$-nonexpansive maps without assuming monotonicity
or access to an explicit game representation.

Smoothed analysis has also yielded polynomial-time guarantees
for turn-based deterministic discounted and mean-payoff games.
Loff and Skomra~\cite{loff2024smoothed} first established such
guarantees for policy iteration on ergodic game graphs. More recently, Asadi, Chatterjee, and
Luo~\cite{asadi2027smoothed} and Loff and
Skomra~\cite{loff2027smoothed} independently obtained
smoothed polynomial-time algorithms for deterministic discounted
and mean-payoff games on arbitrary graphs, removing the
ergodicity assumption.

The study of $\ell_\infty$-nonexpansive maps has a long history
in fixed dimension.
Shellman and Sikorski first obtained an
$O(\log(1/\eps))$-time algorithm in two
dimensions~\cite{shellman2002twodimensional},
and later extended their approach to arbitrary fixed dimension
with running time
$O((\log(1/\eps))^d)$~\cite{shellman2003recursive}.
More recently, Chen, Li, and Yannakakis achieved a quadratic
speedup, reducing the exponent from $d$ to
$\lceil d/2\rceil$~\cite{chen2026quadratic}.
While these algorithms have favorable dependence on the accuracy
for fixed $d$, their running-time bounds remain exponential
in the dimension.
In a different direction, Chen, Li, and Yannakakis established
that $O(d^2\log(1/\eps))$ function evaluations suffice in
arbitrary dimension~\cite{chen2024computing}, later improved to
$O(d\log(1/\eps))$~\cite{chen2025computing}.
The computation required to determine the query points,
however, remained exponential.
Feodorov and Haslebacher subsequently made progress on
this computational bottleneck, obtaining subexponential bounds
on both query complexity and total running
time~\cite{feodorov2026faster}.

A different line of work assumes \emph{monotonicity} with respect to
the coordinatewise order.
For monotone $\ell_\infty$-contractions, Batziou, Fearnley,
Gordon, Mehta, and Savani~\cite{batziou2025monotone}
obtained an $O(\log(1/\eps))$-query algorithm in three dimensions
and extended it to arbitrary dimension through a decomposition
theorem, using
$2^{O(d)}(\log(1/\eps))^{\lceil d/3\rceil}$ queries.
Both algorithms require only polynomial computation between
evaluations, yielding corresponding total running-time bounds
up to polynomial factors.
They also established membership in $\mathsf{UEOPL}$,
with consequences for Shapley's stochastic games.

Without the contraction assumption, Tarski's fixed-point
theorem still guarantees the existence of fixed points for
monotone self-maps of complete lattices~\cite{tarski1955lattice}.
Etessami, Papadimitriou, Rubinstein, and
Yannakakis~\cite{etessami2020tarski} showed that solving
Condon's simple stochastic games and approximating the values
of Shapley's stochastic games reduce to finding Tarski fixed
points, and established connections to equilibria in
supermodular games.
For monotone maps on the finite grid $[n]^d$, where
$[n]=\{1,\ldots,n\}$, Fearnley, Pálvölgyi, and Savani~\cite{fearnley2022tarski}
improved the exponent of $\log n$ in the query bound
from $d$~\cite{dang2011tarski} to
$2\lceil d/3\rceil$. More recently, Chen, Gao, Li, Rubinstein, and
Yannakakis~\cite{chen2026tarski} obtained an
$O(5^d(\log n)^{\lceil\log d\rceil})$-query algorithm,
together with a randomized query lower bound of
$(\log n)^{\Omega(\log d)}$, even for maps with
a unique fixed point.

A parallel line of research studies approximate fixed points of
contractions under other $\ell_p$ norms.
For $\ell_1$-contractions, Fearnley, Gordon, Mehta, and
Savani~\cite{fearnley2020unique} obtained an
$O((\log(1/\eps))^d)$-time algorithm in fixed dimension.
Chen, Li, and Yannakakis~\cite{chen2026quadratic}
recently reduced the exponent from $d$ to $\lceil d/2\rceil$.
Haslebacher~\cite{haslebacher2025arrival} also established
a reduction from generalized \textsc{Arrival}, a deterministic
graph traversal problem, to approximate fixed-point computation
for $\ell_1$-contractions.

For the Euclidean norm ($\ell_2$), Huang, Khachiyan, and
Sikorski~\cite{huang1999approximating} developed an
ellipsoid-based algorithm using $O(d\log(1/\eps))$
function evaluations.
Building on the polynomial query bounds for
$\ell_\infty$-contractions~\cite{chen2024computing},
Haslebacher, Lill, Schnider, and
Weber~\cite{haslebacher2025queryefficient} established
polynomial query complexity for contractions under
every $\ell_p$ norm.
More recently, Daskalakis, Farina, and
Zhang~\cite{daskalakis2026evenp} obtained a
$\operatorname{poly}(d,p,\log(1/\eps))$-time algorithm
for $\ell_p$-nonexpansive maps when $p$ is an even integer,
achieving polynomial total running time in this setting.

\section{Preliminaries and Notation}
\label{sec:preliminaries}

While the function $f$ has domain $[0,1]^d$, our algorithm might sometimes produce search boxes off by a constant amount from the domain. For that, throughout the paper we extend $f$ to $\mathbb{R}^d$ by setting
$f(x)=f(\clip_{[0,1]^d}(x))$ for $x\notin[0,1]^d$.
Since clipping is $\ell_\infty$-nonexpansive, this extension
preserves the contraction gap $\gamma$ and the unique
fixed point $x^*$. For any $N\subseteq[d]$, let $x_N$ denote the subvector containing only the coordinates in $N$,
and let $x_{-N}$ denote its restriction to the complementary
coordinates. Similarly, $f_N(x)$ denotes the restriction of
$f(x)$ to the coordinates in $N$.

We mention a standard lemma that enables converting guarantees between approximate and near fixed points.

\begin{lemma}\label{lem:triv bounds}
For every $x\in\mathbb{R}^d$,
$
    \gamma\norm{x-x^*}
    \leq \norm{x-f(x)}
    \leq (2-\gamma)\norm{x-x^*}.
$
\end{lemma}

\begin{proof}
Write $x - f(x) = (x - x^*) + (x^* - f(x))$ and apply the triangle inequality as well as the contraction property $\|f(x) - x^*\|_\infty \le (1-\gamma)\|x - x^*\|_\infty$. 
\end{proof}

\section{Algorithm Overview}
\label{sec:overview}

Our algorithm maintains a search cube $[L,U] = \prod_{i=1}^d [L_i,U_i]$ of side $U_i - L_i = D$, that contains the true fixed point $x^*$. At each iteration, the algorithm shrinks the length of each side of the cube by a factor $7/8$ by either increasing $L_i$ by $D/8$ or decreasing $U_i$ by $D/8$ along each coordinate $i \in [d]$, while maintaining the invariant that $x^*$ remains in the cube. Hence, within $O(\log(1/\eps))$ iterations, $x^*$ is contained within a box of size $\eps$. The midpoint $q \coloneqq (L + U)/2$ is therefore a $\eps/2$-near fixed point and hence a $\eps$-approximate fixed point by \Cref{lem:triv bounds}.

For every coordinate $i=1,\dots, d$, the algorithm identifies which of the two faces of the cube, at $L_i$ or $U_i$, is safe to push inward by $D/8$. (At least one of these two faces is at distance $\ge D/2$ from $x_i^*$, so at least one face can be safely moved inward by $D/8$). 

To decide which direction to push, we will construct a function $\Label(L, U; D/8, \eps)$ that either returns an $\eps$-near fixed point of $f$, or correctly identifies when $x_i^*$ is within distance $D/8$ from either $L_i$ or $U_i$. In particular, in the latter case $\Label(L, U; D/8, \eps)$ returns labels $\ell_{i,-}$ for each coordinate, with the property that
\[
\numberthis{eq:label main}
    x_i^* \in [L_i, L_i + D/8] \implies \ell_{i,-} = \mathsf{lo};\quad\qquad
    x_i^* \in [U_i - D/8, U_i] \implies \ell_{i,-} = \mathsf{hi}.
\]
Importantly, the label is arbitrary when $x_i^*$ is within $(L_i + D/8, U_i - D/8)$; luckily, in that case it does not matter whether we push $L_i$ or $U_i$ inward. Therefore, it is always safe to update $L_i \gets L_i + D/8$ when $\ell_{i,-} = \mathsf{hi}$ and $U_i \gets U_i - D/8$ when $\ell_{i,-} = \mathsf{lo}$.

Crucially, we will be able to construct such $\Label$ subroutine to run in quasi-polynomial time. We summarize the resulting pseudocode in \cref{alg:solve}.

\begin{algorithm}[htp]
\caption{$\Solve(f,\eps)$ for $d\ge2$}
\label{alg:solve}
\KwIn{A $(1-\gamma)$-contraction $f:[0,1]^d\to[0,1]^d$, and desired tolerance $\eps > 0$.}
\KwOut{An $\eps$-near fixed point of $f$.}
$D\gets1$\;
$L_i\gets0$, $U_i\gets1$ for every $i=1,\dots, d$\;
\While{$D>\eps$}{
 \If(\hfill\tcp*[h]{\cref{algo:label}}){$\Label(L,U; D/8, \eps)$ returns an $\eps$-near fixed point $x$}{
 \KwRet $x$
 }
 \Else{receive labels $\ell_{i,-}$ from the call}
 \ForEach{$i=1,\dots,d$}{
  \lIf{$\ell_{i,-}=\mathsf{lo}$}{$U_i\gets U_i-D/8$}
  \lIf{$\ell_{i,-}=\mathsf{hi}$}{$L_i\gets L_i+D/8$}
 }
 $D\gets7D/8$\;
}
\KwRet $(L+U)/2$\;
\end{algorithm}

For completeness, we define the base cases ($d < 2$) of $\Solve$ separately.
In dimension zero, $\Solve$ returns the empty vector.
In dimension one, $\Solve(f,\eps)$ uses bisection on $f(x)-x$.
Indeed, $f(0)\geq 0$, $f(1)\leq 1$, and $f(x)-x$ is strictly
decreasing by contractivity.
Thus, bisection maintains an interval containing $x^*$,
and returning its midpoint once the interval has length at most
$2\eps$ yields an $\eps$-near fixed point using
$O(1+\log(1/\eps))$ evaluations of $f$.

\subsection{Labelings}

As laid out in \cref{sec:overview}, the key insight behind the algorithm is the fact that it is possible to correctly detect when each coordinate $x_i^*$ is sufficiently close to the extrema of a query box $[L, U]$. The bulk of the algorithm will be focused on how to detect such labels efficiently. Formally, we introduce the following central concept.

\begin{definition}[$\tau$-accurate labeling]\label{def:labeling}
    Let $[L, U] \subset \mathbb{R}^d$ be a box, possibly degenerate (i.e., we allow the possibility that $L_i > U_i$ for some $i$). For each coordinate $i$, we denote the distances from the extrema $L_i$ and $U_i$ to the $i$-th coordinate of the fixed point, $x_i^*$, as
    \[
        \sigma_{i,-} \coloneqq x_i^* - L_i, \qquad\text{and}\qquad \sigma_{i,+} \coloneqq U_i - x_i^*.
    \]
    We refer to these differences as the \emph{slacks} of the faces $L,U$, and refer to them generically with the symbol $\sigma_j$. Let $\sigma_{\max}$ and $\sigma_{\min}$ denote the maximum and minimum values of $\sigma_j$, taken over all $2d$ faces of the box.

    We call assignment of labels $\ell_{i,-},\ell_{i,+}\in\{\mathsf{lo},\mathsf{hi}\}$ to all faces $j \in \{1,\dots,d\}\times\{-,+\}$ \emph{$\tau$-accurate} if it satisfies
    \begin{equation}\label{eq:labeling}
        \sigma_j \ge \sigma_{\max} - \tau \implies \ell_{j}=\mathsf{hi}, \qquad \quad \sigma_j \le \sigma_{\min} + \tau \implies \ell_{j}=\mathsf{lo}.
    \end{equation}
\end{definition}

Some remarks are in order. 

\paragraph{From \Cref{def:labeling} to the guarantee \eqref{eq:label main}.}
First, the following remark shows that our definition of labeling
provides the guarantee required by $\Solve$.

\begin{remark}
We remark that this notion immediately yields the notion of labels described in \eqref{eq:label main}, by setting accuracy $\tau \coloneqq D/8$. Indeed, for any cube $[L,U]$ of side $D$ containing $x^*$, $\sigma_{\min} \ge 0$ and $\sigma_{\max} \le D$, and so any $D/8$-accurate labeling satisfies in particular
\[
    \sigma_{i,-} \le \frac{D}{8} \implies \ell_{i,-} = \mathsf{lo}, \qquad \sigma_{i,-} \ge \frac{7}{8}D \implies \ell_{i,-} = \mathsf{hi}.
\]
The first condition is precisely $x_i^* \in [L_i, L_i + D/8]$; the second condition is $x_i^* \in [L_i + 7D/8, U_i] = [U_i - D/8, U_i]$. Hence, any $D/8$-accurate labeling implies \eqref{eq:label main}.
\end{remark}

\paragraph{Non-existence of labelings.}
Second, we point out that a valid $\tau$-accurate labeling may not exist. As the following remark shows, in this case the center of the box is already a $\tau$-near fixed point.
\begin{remark}\label{rem:center}
    It is possible that a labeling as in \Cref{def:labeling} might not exist due to incompatible conditions in \eqref{eq:labeling}. Such inconsistency is possible only when
    $
        \sigma_{\max} - \tau \le \sigma_{\min} +  \tau,
    $
    that is $\sigma_{\max} - \sigma_{\min} \le 2\tau$. But then, this immediately certifies that the center $q \coloneqq (L+U)/2$ satisfies
    \[
        2 |x_i^* - q_i| = |\sigma_{i,-} - \sigma_{i,+}| \le \sigma_{\max} - \sigma_{\min} \le 2\tau \quad\qquad \forall i = 1,\dots, d.
    \]
    Hence, non-existence of a labeling immediately certifies that the center $q$ is a $\tau$-near fixed point.    
\end{remark}

\paragraph{Degenerate boxes.} Finally, we remark that allowing for degenerate boxes is always without loss of generality.

\begin{remark}\label{rem:degen}Although picturing $[L,U]$ as a nonempty box captures the intended application, internally our algorithm might end up querying a labeling for a \emph{degenerate} box $L,U$ (i.e., one that does not respect $L_i < U_i$ for all $i$). For the purposes of labelings, the extension of labelings to degenerate boxes is without loss of generality: every time $[L,U]$ is degenerate, one can find $c \in \mathbb{R}$ such that the repaired box $[L - c\mathbf{1}, U+c\mathbf{1}]$ is non-degenerate. Crucially, the labeling guarantees for such a repaired box are exactly the same guarantees as \cref{def:labeling}, since adding or subtracting $c$ from all entries of $L$ and $U$ shifts by the same quantity all $\sigma_j$ and thus also $\sigma_{\min}$ and $\sigma_{\max}$. To keep the exposition as simple as possible and avoid introducing that additional step, we opted to introduce \cref{def:labeling} in the more permissive form directly.
\end{remark}

\subsection{Overview of Labeling Construction}

In light of \cref{rem:center}, our goal will be to construct a generic routine $\Label(L,U; \tau, \eps)$ whose output is either a $\tau$-accurate labeling or a $\eps$-near fixed point.

We will show that a $\tau$-accurate label can be computed in time $\log(\tau/(\gamma \eps))^{O(\log^2 d)}$. The construction uses two steps. In the first step, we show that it is possible to construct a labeling that is very little accurate (specifically, $\tau = O(\gamma)$) by splitting coordinates based on the median length of each side, recursively solve half-dimensional contractions, and combine the answers. Such labelings are correct only for faces for which $x^*_i$ is extremely close; in other words, they detect proximity to the fixed point only on a \emph{thin} inset of the search box. We call such a labeling of accuracy $O(\gamma)$ a \emph{thin} labeling. To amplify the accuracy of the labeling, in the second step we use a \emph{boosting} procedure (based on classic AdaBoost). In particular, we query several thin labelings and perform majority voting to construct final labels.

\section{Constructing Thin Labelings}\label{sec:thin}

As the starting point for the entire construction, we will show that minimally-accurate labelings can be constructed efficiently. In particular, we will construct a labeler that, given a (possibly degenerate; see also \cref{rem:degen}) box $[L, U]$ and any $\eta > 0$, either:
\begin{itemize}
    \item returns a point $p \in [0,1]^d$ (not necessarily in the box) that is an $\eta$-approximate fixed point; or
    \item returns a $\frac{\gamma\eta}{8}$-accurate labeling for $[L,U]$.
\end{itemize}

The algorithm sorts the coordinates $i=1,\dots, d$ by increasing side length $U_i - L_i$. Letting $m$ denote the \emph{median} side length, it then splits the dimensions between those that are narrow and those that are wide, that is,
\[
    N \coloneqq \{i : U_i - L_i < m\}, \qquad \quad W \coloneqq \{i: U_i - L_i > m\}.
\]
Let $q\coloneqq(L+U)/2$ be the center of the box.
Recall from \Cref{sec:preliminaries} that $f_N$ denotes the restriction of the output of $f$
to the coordinates in $N$, and similarly for $f_W$.
The algorithm recursively invokes $\Solve$ at precision $\gamma\eta/2$
on the maps $y\mapsto f_N(y,q_{-N})$ and
$y\mapsto f_W(q_{-W},y)$.
Both maps are $(1-\gamma)$-contractions on their unit cubes,
since fixing input coordinates and restricting output coordinates
preserve the contraction property.
Each recursive problem has dimension at most $d/2$.

\begin{algorithm}[t]
\caption{$\ThinLabel(L,U; \eta)$: constructing thin labels}
\label{algo:thin}
\KwIn{A possibly degenerate box $[L,U]$, and a tolerance $\eta>0$.}
\KwOut{A $\frac{\gamma\eta}{8}$-accurate labeling for $[L,U]$, or an $\eta$-approximate fixed point $p \in [0,1]^d$.}
$w\gets (U-L)/2$\tcp*[r]{Half side lengths}
$t\leftarrow\text{median of the } \{w_i\}_{i=1}^d$\;
$N\leftarrow\{i\in\{1,\dots,d\}:w_i<t\}$\tcp*[r]{``Narrow'' dimensions}
$W\leftarrow\{i\in\{1,\dots,d\}:w_i>t\}$\tcp*[r]{``Wide'' dimensions}
\BlankLine
\tcp{Part I: Construct candidate point $x$ by recursion}
$q\gets(L+U)/2$\tcp*[r]{Center of box}
$y^\text{narr} \in \mathbb{R}^N \gets \Solve(y \mapsto f_N(y, q_{-N}), \gamma\eta/2)$\nllabel{line:narrow}\tcp*[r]{Fix wide, solve for narrow coords.}
$y^\text{wide} \in \mathbb{R}^W \gets \Solve(y \mapsto f_W(q_{-W}, y), \gamma\eta/2)$\tcp*[r]{Fix narrow, solve for wide coords.}
Define $x$ as $\begin{cases}
    \clip_{[U_i - t, L_i + t]}(y^\text{narr}_i) & i \in N \;\\
    \clip_{[L_i + t, U_i - t]}(y^\text{wide}_i) & i \in W \;\\
    q_i & \text{otherwise}
\end{cases}$\nllabel{line:clip}\;
\BlankLine
\tcp{Part II: Inspect $x$}
$h\leftarrow f(x)$\;
\If{$\norm{x-h}\le\eta$}{\textbf{return} $\eta$-approximate fixed-point $\clip_{[0,1]^d}(x)$\nllabel{line:defeat}\;}
\Else{
\ForEach{$i=1,\dots,d$}{
    $\ell_{i,-} \gets \mathsf{hi} \text{ if } h_i - L_i \ge t, \text { else } \mathsf{lo}$\nllabel{line:thin-lo}\;
    $\ell_{i,+} \gets \mathsf{hi} \text{ if } U_i - h_i \ge t, \text { else } \mathsf{lo}$\nllabel{line:thin-hi}\;
}
\KwRet labels $(\ell_{-}, \ell_{+})$\;
}
\end{algorithm}

\begin{lemma}
For any possibly degenerate box $[L,U]$ and $\eta>0$,
$\ThinLabel(L,U;\eta)$ (\Cref{algo:thin}) returns either an $\eta$-approximate
fixed point in $[0,1]^d$ or a $\frac{\gamma\eta}{8}$-accurate
labeling of $[L,U]$.
\end{lemma}

\begin{proof}
    If $\|x-f(x)\|_\infty\leq\eta$, the algorithm returns
    $p\coloneqq\operatorname{clip}_{[0,1]^d}(x)$.
    By the definition of the clipped extension,
    $f(p)=f(x)\in[0,1]^d$.
    Since clipping is nonexpansive, we have
    \[
    \|p-f(p)\|_\infty
    =\|\operatorname{clip}_{[0,1]^d}(x)
    -\operatorname{clip}_{[0,1]^d}(f(x))\|_\infty
    \leq\|x-f(x)\|_\infty\leq\eta.
    \]
    Thus, $p$ is an $\eta$-approximate fixed point.

    So, consider the case $\|x - f(x)\|_\infty > \eta$, which implies $\|x - x^*\|_\infty > \eta/2$ by \cref{lem:triv bounds}.    
    The crux of the proof is to show that when that is the case, $\sigma_{\min}$ and $\sigma_{\max}$ are sufficiently separated from the median half-side $t$, specifically
    \begin{equation}\label{eq:thin crux}
        \sigma_{\min} \le t - \|x - x^*\|_\infty \qquad \text{and}\qquad \sigma_{\max} \ge t + \|x - x^*\|_\infty.
    \end{equation}
    If we \emph{estimate} the distances from $x^*$ to the faces by pretending that $h \coloneqq f(x) \approx x^*$, i.e.,
    \[
        \hat{\sigma}_{i,-} \coloneqq h_i - L_i, \qquad \hat{\sigma}_{i,+} \coloneqq U_i - h_i,
    \]
    the errors $|\hat{\sigma}_j - \sigma_j|$ would be off by at most $\|h - x^*\|_\infty \le (1-\gamma)\|x - x^*\|_\infty$. In light of \eqref{eq:thin crux}, this error is less than the distance between the $t$ and the extrema $\sigma_{\min}, \sigma_{\max}$, meaning that values of $\sigma_j$ very close to $\sigma_{\min}$ and $\sigma_{\max}$ will still land to the left and right of $t$ even in the unlucky case that their estimates achieve the maximum error. In particular,
    \begin{align*}
        \sigma_j - \sigma_{\min} \le \frac{\gamma\eta}{8} \implies \hat\sigma_j &\le \sigma_{\min} + \frac{\gamma\eta}{8} + (1-\gamma)\|x - x^*\|_\infty\\
        &\le t - \gamma\|x - x^*\|_\infty + \frac{\gamma\eta}{8}< t,
    \end{align*}
    where we used \eqref{eq:thin crux} in the second inequality, and the fact that $\|x-x^*\|_\infty > \eta/2$ by assumption in the third.
    Symmetrically,
    \[
    \sigma_{\max} - \sigma_j  \le \frac{\gamma\eta}{8} \implies \hat\sigma_j > t.
    \]
    This shows that the assignments in \texttt{Part II} (\cref{line:thin-lo,line:thin-hi}) satisfy the contract in \eqref{eq:labeling}.

    It remains to show that \eqref{eq:thin crux} holds. Throughout the rest of the proof, \ul{let $i$ be the coordinate for which $|x_i - x^*_i| = \|x - x^*\|_\infty$}. We distinguish three cases:
    \begin{enumerate}
        \item If $i$ is neither narrow nor wide, then the $i$-th side length is exactly $2t$ and $x_i$ is the midpoint of $[L_i, U_i]$. Hence, the fixed point $x^*$ sits at distances $t \pm \|x - x^*\|_\infty$ from the faces $L_i, U_i$, and \eqref{eq:thin crux} follows.
        \item If $i$ is narrow, the clipping on \cref{line:clip} applies, and in particular $x_i \in [U_i - t, L_i + t]$. Hence, $x_i$ is at distance \emph{at most} $t$ from both $L_i$ and $U_i$. In moving from $x_i$ to $x_i^*$, at least one of the slacks must decrease by $\|x-x^*\|_\infty$, showing that 
        \[\sigma_{\min} \le t - \|x-x^*\|_\infty.\]
        To prove the remaining direction, $\sigma_{\max} \ge t + \|x-x^*\|_\infty$, we will argue about the existence of some face $j$ such that $\sigma_j \ge t + \|x-x^*\|_\infty$. This is where the fact that $y^\text{narr}$ is defined as the (approximate) solution of a smaller-dimension fixed-point problem (\cref{line:narrow}) becomes crucial.
        \begin{enumerate}[label={\roman*)}]
            \item As a first case, suppose $|y^\text{narr}_{i} - x^*_i| \ge |x_i - x^*_i| = \|x - x^*\|_\infty$. Let $y^*_{i}$ be defined as the fixed point of the restriction $f_N(\cdot, q_{-N})$ where the non-narrow coordinates are set to the box center. Since $y^*$ is a fixed point, we must have
            \begin{align*}
                \|y^* -x_N^*\|_\infty=\max_{i \in N} |y^*_i - x^*_i| &= 
                \max_{i \in N} |f_{N,i}(y^*, q_{-N}) - f_i(x^*)| \\
                &\le 
                \|
                    f_N(y^*, q_{-N}) - f_N(x^*)
                \|_\infty \\
                &\le (1-\gamma) \max\{\|y^* -x_N^*\|_\infty, \|q_{-N} - x^*_{-N}\|_\infty\},
            \end{align*}
            where the second equality used the definition of $y^*$ as a fixed point of $f_N(\cdot, q_{-N})$, and the last inequality used the fact that $f$ is a $(1-\gamma)$-contraction. Crucially, the maximum on the right cannot be $\|y^* -x_N^*\|_\infty$ as if $\|y^* -x_N^*\|_\infty > \|q_{-N} - x^*_{-N}\|_\infty \ge 0$, then we would have $\|y^* -x_N^*\|_\infty \le (1-\gamma) \|y^* -x_N^*\|_\infty$, impossible. Hence, the error $\|y^* -x_N^*\|_\infty$ is bounded by
            \[
                \|y^* -x_N^*\|_\infty \le (1-\gamma)\|q_{-N} - x^*_{-N}\|_\infty.
            \]
            Accounting for the fact that $y^\text{narr}$ is an $\gamma\eta/2$-near fixed point, we obtain
            \[
                \|y^\text{narr} -x_N^*\|_\infty \le (1-\gamma)\|q_{-N} - x^*_{-N}\|_\infty +\frac{\gamma\eta}{2}.
            \]

            Since $|y^\text{narr}_{i}  - x^*_i| \ge \|x - x^*\|_\infty > \eta/2$ by assumption, the error bound in the previous line implies $\|q_{-N} - x^*_{-N}\|_\infty > \|x - x^*\|_\infty$.\footnote{Or else, we would have the contradiction 
                \[
                     \|x - x^*\|_\infty \le \|y^\text{narr} -x_N^*\|_\infty \le  (1-\gamma)\|x - x^*\|_\infty + \frac{\gamma\eta}{2} \quad\implies\quad \|x-x^*\|_\infty \le \frac{\eta}{2}.
                \]
            }
            But then, there must exist a non-narrow coordinate $j$ (i.e., $w_j \ge t$) for which 
            \[
                |q_j - x^*_j| > \|x - x^*\|_\infty.
            \]
            Hence, at least one of the slacks $x_j^*-L_j$ and $U_j-x_j^*$
            must exceed $w_j+\|x-x^*\|_\infty$.
            Since $w_j\geq t$, we conclude that
            $\sigma_{\max}\geq\sigma_j>t+\|x-x^*\|_\infty$,
            as we wanted to show.

            \item Conversely, assume $|y^\text{narr}_{i} - x^*_i| < |x_i - x^*_i|$, that is, clipping $y^\text{narr}_{i}$ onto $[U_i - t, L_i + t]$ on \cref{line:clip} has \emph{increased} the distance to $x^*_i$. This can only happen if $y^\text{narr}_{i}$ and $x_i^*$ are outside of the interval, on the same side. 

            \vspace{2mm}
            If $y^\text{narr}_{i}, x_i^* \le U_i - t$, the clipping will force $x_i = U_i - t$. Since $x_i^*$ needs to be on the left, this means that the slack of the face at $U_i$ is
            \[
                \sigma_{i,+} = U_i - x_i^* = U_i - x_i + |x_i - x_i^*| = t + \|x - x^*\|_\infty.
            \]
            The case $y^\text{narr}_{i}, x_i^* \ge L_i + t$ is symmetric, and shows that $\sigma_{i,-} \ge t + \|x - x^*\|_\infty$.
        \end{enumerate}
         In both cases, $\sigma_{\max} \ge t + \|x - x^*\|_\infty$, concluding the analysis.
        \item Symmetrically, if $i$ is wide, the clipping on \cref{line:clip} ensures that $x_i$ is at distance at \emph{at least} $t$ from both $L_i$ and $U_i$, showing that
        \[\sigma_{\max} \ge t + \|x-x^*\|_\infty.\]
        To show that $\sigma_{\min} \le t - \|x-x^*\|_\infty$, we can follow the same steps as in the narrow case, this time seeking a face $j$ such that $\sigma_j \le t - \|x-x^*\|_\infty$.\qedhere
    \end{enumerate}
\end{proof}

\section{Amplifying Thin Labelings}

The labeler in \cref{sec:thin} guarantees returning correct labels $\{\mathsf{lo},\mathsf{hi}\}$ only for those faces whose slacks from the fixed point $x^*$ are within distance $\gamma\eta$ from the minimum and maximum slack. Unfortunately, $\gamma\eta$ is extremely small when $\gamma$ is. In this section we can show that we can amplify the accuracy of the labelings from $\gamma\eta$ to constant paying a price that is quasipolynomial in the \emph{logarithm} of $1/\gamma\eta$. To do so, we will use an idea similar to AdaBoost~\cite{freund1997decision,freund1996experiments}.

The technique we use is general and does not rely in any way on $f$ or any geometric property of the problem. Hence, to fix ideas and simplify the treatment, we suggest that the reader lets go of any geometric aspect, and abstract the question as follows.

\begin{problem}[Extremal amplification game]
\label{prob:ampli}
Consider the following game. An adversarial $\Oracle$ holds a secret vector $\sigma \in \mathbb{R}^n$ of size $n$. We can interact with the $\Oracle$ by proposing a vector of shifts $(b_i)$ for each of the coordinates. Given the vector $b$, $\Oracle(b)$ can either:
\begin{enumerate}
    \item declare defeat; or
    \item return a labeling $\{-1,+1\}^n$ that must tag with $-1$ all the coordinates $i$ such that $\sigma_i + b_i$ is within distance $\rho$ of the minimum of $\sigma+b$, and as $+1$ all those within distance $\rho$ of the maximum of $\sigma+b$; all other coordinates can be labeled either way.
\end{enumerate} 
Our objective in the game is to have the $\Oracle$ declare defeat, or be able to extract a labeling for $\sigma$ that is $\kappa$-accurate, that is, tags as $-1$ those coordinates that are within $\kappa$ of the minimum of $\sigma$, and as $+1$ those that are within $\kappa$ of the maximum of $\sigma$.
\end{problem}

\begin{remark}\label{rem:remap}
\cref{prob:ampli} remaps directly onto the thin labeler of \cref{sec:thin}. The vector $\sigma$ is precisely the vector of face slacks $\sigma_{j,+}$ and $\sigma_{j,-}$ across the $n=2d$ faces of the (possibly degenerate) box $[L,U]$. Similarly, shifting the values of the face slacks $\sigma_{j,\pm}$ is achieved by adjusting the faces $L, U$ passed to $\ThinLabel$ by the $b$.\footnote{Specifically, if the coordinate $i$ corresponds to face $(j,+)$, $\sigma_i + b_i$ can be induced by moving $U_{j}$ up by $b_i$. If coordinate $i$ corresponds to face $(j,-)$, $\sigma_i + b_i$ is induced by moving $L_j$ down by $b_i$.}
The labels $-1$ and $+1$ of \cref{prob:ampli} are precisely $\mathsf{lo}$ and $\mathsf{hi}$ returned by \cref{algo:thin}; we use a numerical identifier since we will be keeping vote tallies to perform boosting, and thus we will find it convenient to refer to sums of labels. The threshold $\gamma\eta/8$ guaranteed by \cref{algo:thin} corresponds to $\rho$, while $\kappa$ is the desired threshold $D/8$ required by \eqref{eq:label main}. Finally, ``defeat'' corresponds to the case in which \cref{algo:thin} returns an $\eta$-approximate fixed point (\cref{line:defeat}). 
\end{remark}

\paragraph{A naive attempt.} A correct (but too slow) approach to solve the labeling amplification game would be the following.\footnote{We ignore the possibility of the oracle admitting defeat. If at any point the oracle admits defeat, the game is immediately over in our favor.} First, query $\Oracle(0)$. This call must return labels such that all elements in the range $[\sigma_{\min}, \sigma_{\min} + \rho]$ get a $-1$, and all those in the range $[\sigma_{\max}-\rho,\sigma_{\max}]$ get a $+1$. Hence, we can now shift each $i$ that got a $-1$ label \emph{up} by the amount $b_i = \rho$. Given the new labeling, we can symmetrically shift each $i$ that got a $+1$ label \emph{down} by $b_i = \rho$. It is easy to see that by the end, the majority voting of the three labelings returned by $\Oracle$ is a valid $2\rho$-accurate labeling. We could then keep iterating the idea until accuracy $\kappa$ is reached. Such an approach would then solve the game using a number of calls proportional to $\kappa / \rho$, which is too slow for our purposes as $\rho \le \gamma$.

\paragraph{Main result.}
Instead, the main result we present in this section is a boosting approach for solving the extremal amplification game that depends quasi-polynomially on the \emph{logarithm} of the ratio of accuracies. In particular, we will seek to establish the following guarantee.

\begin{theorem}\label{thm:amplify cost}
For any $n\ge2$ and $\kappa \ge 4\rho > 0$, there exists a deterministic procedure
$\Amplify(0,2\kappa,0,n,n)$ that solves the extremal amplification game (\cref{prob:ampli}) and returns either an admission of defeat
from $\Oracle$, or a $\kappa$-accurate labeling for $\sigma$.
It performs at most
\[
    Q \coloneqq \log(\kappa/\rho)^{C \log n}
\]
calls to $\Oracle$ and $O(nQ)$ additional arithmetic operations,\footref{fn:bits} where $C$ is a universal constant.
\end{theorem}

As we elaborate below, the main idea of the algorithm is based on boosting~\cite{freund1996game,freund1997decision}.

\subsection{Vanilla Boosting}\label{sec:boosting}

The amplification process in our algorithm draws on boosting ideas. To settle notation and keep the presentation as self-contained as possible, we now recall the key idea of the AdaBoost algorithm.

Suppose we have $n$ examples, each with a known label $y_i \in \{-1,+1\}$. An imperfect oracle (typically referred to as a weak learner) receives as input a distribution $p_i$ over the examples and produces a vector of labels $\ell \in \{-1,+1\}^n$ with the property that the total mass of correct predictions is greater than one half; to keep things concrete, let's say at least $3/4$. 

The main statement of boosting is that it is then possible, by calling the oracle a \emph{logarithmic} in $n$ number of times, to recover correct labels for all examples by majority vote on the labels produced by the oracle.

The possibility of boosting is a reflection of the minimax theorem, which can be made constructive via no-regret learning~\cite{freund1996game}. When using multiplicative weights update as the no-regret algorithm, the resulting algorithm is termed AdaBoost~\cite{freund1997decision}. We now present a standalone, direct potential argument for the correctness of the resulting algorithm, avoiding the deeper connection to the minimax theorem. 

For each example $i\in\{1,\dots,n\}$, AdaBoost keeps a vote tally $F_i$ of the proposed label so far, and considers the potential function defined as the sum of exponential weights
\[
    \L(F) \coloneqq \sum_{i: y_i = +1} \sqrt{3}^{-F_i} + \sum_{i:y_i=-1} \sqrt{3}^{F_i}.
\]
In particular, when the potential is $< 1$, $F_i$ must be $<0$ whenever $y_i=-1$, and $F_i > 0$ whenever $y_i=+1$; in other words, the majority vote must correctly label all of the examples. At the beginning, $F_i=0$ for all examples, so the potential starts from $n$; as soon as the potential drops below $1$, the majority vote is accurate.

At each iteration, AdaBoost queries the imperfect oracle on the distribution
\[
    p_i = \begin{cases}
        \frac{\sqrt{3}^{-F_i}}{\L(F)} & \text{if } y_i = +1\\
        \frac{\sqrt{3}^{F_i}}{\L(F)} & \text{if } y_i = -1,
    \end{cases}
\]
where the probabilities are assigned in proportion to the contribution of each example to $\L(F)$. Let $E$ be the total exponential weight of all examples that the imperfect oracle predicts incorrectly, 
\[
E \coloneqq \sum_{\substack{i:y_i=+1, \ell_i=-1}}\sqrt3^{-F_i} + \sum_{\substack{i:y_i=-1,\ell_i=+1}}\sqrt3^{F_i}.
\]
Then, all of those will increase by a factor $\sqrt{3}$, whereas the other examples will see their exponential weight go down by the factor $1/\sqrt 3$. Hence, letting $F' \coloneqq F + \ell$ denote the updated tallies, the potential satisfies
\begin{equation}\label{eq:boosting}
    \L(F') = \frac{1}{\sqrt 3}(\L (F) - E) + \sqrt{3} E = \frac{1}{\sqrt{3}} \L(F) + \left(\sqrt{3} - \frac{1}{\sqrt 3}\right) E.
\end{equation}
Crucially, given the way the $p_i$ are defined, the total mass of the misclassified examples is exactly $E / \L(F)$, and is at most $1/4$ due to the guarantee of the imperfect oracle. Hence, substituting $E \le \frac{1}{4} \L(F)$ we obtain
\[
    \L(F') \le \left[\frac{1}{\sqrt{3}} + \frac{1}{4}\left(\sqrt{3} - \frac{1}{\sqrt 3}\right)\right]\L(F) = \frac{\sqrt 3}{2} \L(F).
\]
Since $\frac{\sqrt 3}{2} < 1$, this means that within $O(\log n)$ rounds the potential drops below 1, and thus the majority vote correctly labels all examples.%

At its core, our algorithm follows a potential analysis inspired by the above, with a few twists. A major source of complexity arises however from the fact that the labels $y_i$ are not available in our case. Instead, we leverage the following observation, which follows directly from \eqref{eq:boosting}: provided an upper bound $B$ on the potential $\L(F)$ is known, then
\begin{equation}\label{eq:boosting relaxed}
    \L(F) \le B, \quad E \le B/4 \qquad\implies\qquad \L(F') \le \frac{B}{\sqrt{3}} + \left(\sqrt3 - \frac{1}{\sqrt3}\right)\frac{B}4 = \frac{\sqrt3}{2} B.
\end{equation}
Hence, as long as we can prove that the wrong predictions account for at most one quarter of (the bound on) the potential, it is safe to reduce the bound by a constant fraction and continue the algorithm.

\subsection{Algorithm}\label{sec:ampli}

Fix a shift $b$ and a desired tolerance $\tau$.
Our algorithm, $\Amplify$ (\Cref{algo:amplify}), runs two boosting processes---one for coordinates near the maximum and one for coordinates near the minimum of $\sigma + b$---to amplify the labeling accuracy for $\sigma + b$ up to $\tau$. To manufacture the condition \eqref{eq:boosting relaxed} that enables boosting, we construct a recursive algorithm with the following property:
\begin{itemize}
    \item The algorithm is guaranteed to correctly label coordinates directly when the exponential weight of the examples at the extrema is \emph{low};
    \item If the previous bullet point fails, the exponential weight of the coordinates close to the extrema must be large, and in that case a \emph{single} recursive call to $\Oracle$ (at half accuracy) correctly classifies a large portion of the exponential weight, guaranteeing the condition $E \le B/4$ required in \eqref{eq:boosting relaxed}.
\end{itemize}

In other words, our algorithm gives guarantees \emph{conditionally} on the (bound on) exponential weight of examples near the extrema. To make this quantitative, introduce the symbols\footnote{To ease the notational burden, we avoid explicitly associating the symbols with the particular choice of $b$ and $\tau$. The specific $b$ and $\tau$ inducing the $\sigma', H', L',$ and bands is either generic, or clear from context.}
\begin{equation}\label{eq:sigma H L}
    \sigma'_i \coloneqq \sigma_i + b_i,\qquad H' \coloneqq \max_i \sigma'_i,\qquad L' \coloneqq \min_i \sigma'_i,
\end{equation}
we introduce the extremal bands
\begin{align*}
    R_+ &\coloneqq \{i : \sigma'_i > H' - \tau\},&  S_+ &\coloneqq \{i : \sigma'_i > H' - 8\tau\},\\
    R_- &\coloneqq \{i : \sigma'_i < L' + \tau\},& S_- &\coloneqq \{i : \sigma'_i < L' + 8\tau\},
\end{align*}
and the corresponding exponential potentials as a function of the votes $F_i$ on each coordinate
\[
    \Phi_+(F) \coloneqq \sum_{i \in S_+} \sqrt{3}^{-F_i}, \hspace{2.5cm} \Phi_-(F) \coloneqq \sum_{i \in S_-} \sqrt{3}^{F_i}.
\]
The sets $R_+$ and $R_-$ are exactly the coordinates that must be correctly labeled in the extremal amplification game. However, we define the exponential weight potentials on the larger sets $S_+, S_-$, to account for the shifts applied in the recursive calls.

\begin{algorithm}[!h]
\SetKwInput{KwConvention}{Convention}
\caption{$\Amplify(b,\tau, F, B_+, B_-)$}
\label{algo:amplify}
\KwIn{Shift vector $b$, tolerance $\tau$, current vote tally $F$, potential bounds $B_\pm>0$ for $\Phi_\pm(F)$\!\!\!\!\!\!\!\!}
\KwOut{A defeat from $\Oracle$, or a $\tau$-accurate labeling consistent with \cref{contract:main}}
\KwConvention{If any call returns an admission of defeat from $\Oracle$, immediately propagate that up the call tree}
\BlankLine
\lIf*{$\tau \le \rho$}{\KwRet $\Oracle(b)$\nllabel{line:base-oracle}\tcp*[r]{Base case}}
\lIf*{$B_+B_- < 1$}{\KwRet labels according to \eqref{eq:decoding}\nllabel{line:base-decode}\;}
$\ell_i \gets \bot$ for all $i=1,\dots, n$\tcp*[r]{Output labels unset}
\While(\hfill\tcp*[h]{Perform boosting iterations}\nllabel{line:while}){$B_+ B_- \ge 1$\nllabel{line:boost-loop}}{
    $\delta \gets 0$\nllabel{line:boost-init}\;
    \tcp{Step I: Push up coordinates with label -1}
    \ForEach{$t=1,\dots,64$\nllabel{line:step1-loop}}{
        $\ell' \gets \Amplify(b + \delta, \frac{\tau}{32}, F, B_+, \frac34 B_-)$\nllabel{line:step1-call}\;\vspace{2mm}
        $\delta_i \gets \delta_i + \frac{\tau}{32} \mathbf{1}[\ell'_i = -1] $ for all $i = 1,\dots, n$\nllabel{line:shift-up}
    }
    \tcp{Step II: Push down coordinates with label +1}
    \ForEach{$t=1,\dots,64$}{
        $\ell' \gets \Amplify(b + \delta, \frac{\tau}{32}, F, \frac34 B_+, B_-)$\nllabel{line:step2-call}\;\vspace{2mm}
        $\delta_i \gets \delta_i - \frac{\tau}{32} \mathbf{1}[\ell'_i = +1] $ for all $i = 1,\dots, n$\nllabel{line:shift-down}
    }
    \BlankLine
    \tcp{Step III: Set some labels directly if possible}
    \ForEach{$i=1,\dots,n$\nllabel{line:direct-labels}}{
        \lIf{$\ell_i = \bot$ and $\delta_i \ge \tau$}{$\ell_i \gets -1$\nllabel{line:label-minus}}
        \lIf{$\ell_i = \bot$ and $\delta_i \le -\tau$}{$\ell_i \gets +1$\nllabel{line:label-plus}}
    }
    \BlankLine
    \tcp{Step IV: Recurse *once*, at half precision; update vote tallies}
    $\ell' \gets \Amplify(b + \delta, \frac{\tau}{2}, F, B_+, B_-)$\nllabel{line:weak-learner}\tcp*[r]{Votes from weak learner}
    \ForEach{$i=1,\dots,n$}{
        $F_i \gets F_i + \ell'_i$\nllabel{line:vote-update}\tcp*[r]{Update vote tally}
    }
    \tcp{Update potential bounds according to \eqref{eq:boosting relaxed}}
    $B_+ \gets \frac{\sqrt 3}{2}B_+$,\quad
    $B_- \gets \frac{\sqrt 3}{2}B_-$\nllabel{line:potential-update}\;
}
Decode unset ($\bot$) labels via \eqref{eq:decoding} and return all labels $\ell_i$\;
\end{algorithm}

We consider and maintain upper bounds $B_+$ and $B_-$ on the two potentials,
satisfying $\Phi_+(F)\leq B_+$ and $\Phi_-(F)\leq B_-$, respectively. As with AdaBoost, when $B_+$ and $B_-$ fall under $1$, the signs of $F$ provide correct labels. In fact, a more relaxed termination condition, looking at the \emph{product} $B_+ B_-$ rather than each individual bound, can be derived.
\begin{lemma}\label{lem:decoding}
    Assume that $B_+ B_- < 1$, and consider the decoding rule for labels
    \begin{equation}\label{eq:decoding}
        \ell_i \coloneqq \begin{cases}
            +1 & \text{if } 3^{F_i} \ge B_- / B_+ ; \\
            -1 & \text{otherwise}.
        \end{cases}
    \end{equation}
    If $\Phi_+(F) \le B_+$, the labels assign $+1$ to all $i \in S_+$; similarly, if $\Phi_-(F) \le B_-$, the labels assign $-1$ to all $i \in S_-$.
\end{lemma}
\begin{proof}
    Expanding the definition of $\Phi_+(F)$ in the hypothesis $\Phi_+(F) \le B_+$, we must have, for all $i \in S_+$,
    \[
        \sqrt{3}^{-F_i} \le B_+ \implies 3^{-F_i} \le B_+^2 \implies 3^{F_i} \ge B_+^{-2} > B_- / B_+.
    \]
    Similarly, for all $i \in S_-$,
    \[
        \sqrt{3}^{F_i} \le B_- \implies 3^{F_i} \le B_-^2 < B_- / B_+.
    \]
\end{proof}

Each call to $\Amplify$ supplies as state the shift vector $b$, the desired $\tau$, the current tally of votes $F$, and bounds $B_+, B_-$ on the potentials $\Phi_+$ and $\Phi_-$. We give the pseudocode in \cref{algo:amplify}. The relaxed termination condition $B_+ B_- < 1$ for the boosting is exactly what is used in \cref{line:base-decode,line:while}.

\subsection{Correctness}
We prove the correctness of $\Amplify$ by relying on the following inductive contract.
\begin{contract}\label{contract:main}
If the quantity $B_\pm$ indeed bounds the potential $\Phi_\pm$, then the call to $\Amplify(b, \tau, F, B_+, B_-)$ returns valid labels $\ell_i$ for the coordinates in $R_\pm$. More specifically,
\begin{align*}
    \Phi_+(F) \le B_+ &\implies \ell_i = +1 \quad\forall i \in R_+;\\
    \Phi_-(F) \le B_- &\implies \ell_i = -1 \quad\forall i \in R_-.
\end{align*}
\end{contract}

In particular, we will seek to establish the following key result about \cref{algo:amplify}.
\begin{theorem}\label{thm:amplify}
For any shift $b$, margin $\tau>0$, votes $F$, and positive
budgets $B_+,B_-$, the procedure
$\Amplify(b,\tau,F,B_+,B_-)$ terminates and returns either
an admission of defeat from $\Oracle$, or a labeling satisfying \cref{contract:main}.
\end{theorem}

To solve the extremal amplification game from \Cref{prob:ampli}, it therefore suffices to call $\Amplify$ with no shift, desired accuracy $\kappa$, and no vote tally. In that case, the exponential weight potentials are guaranteed to be $B_+ = B_- = n$, leading to the following direct corollary.

\begin{corollary}
    A solution to the extremal amplification game from \Cref{prob:ampli} is obtained by
    $\Amplify(0,2\kappa,0,n,n)$.\footnote{The use of $2\tau$ instead of $\tau$ is due to the fact that $R_+$ and $R_-$ are defined via strict inequalities, while the notion of accuracy is defined in terms of closed intervals.}
    In particular, for a (possibly degenerate) box $[L,U]$ in
    dimension $d$, set $n=2d$ and implement $\Oracle$ using
    $\ThinLabel$ with $\eta=\gamma\eps$, so that
    $\rho=\gamma\eta/8=\gamma^2\eps/8$.
    
    In particular, $\Amplify(0,2\tau,0,n,n)$ with $n=2d$ yields either an
    $\eps$-near fixed point or a $\tau$-accurate labeling
    of $[L,U]$, giving the algorithm
    $\Label(L,U;\tau,\eps)$ in \cref{algo:label}.
\end{corollary}

\begin{algorithm}[!htp!]
\caption{$\Label(L, U; \tau, \eps)$}
\label{algo:label}
    \KwIn{Possibly degenerate box $[L, U]$, desired labeling accuracy $\tau$ and nearness threshold $\eps$}
    \KwOut{Either an $\eps$-near fixed point of $f$, or a $\tau$-accurate labeling for $[L,U]$}
    Have $\Oracle(b)$ be provided by $\ThinLabel$ (\cref{rem:remap}) with desired precision $\eta = \gamma\eps$\;
    Set $n \leftarrow 2d$\;
    Call $\Amplify(0, 2\tau, 0, n, n)$\;
    \If{the call results in the oracle admitting defeat (i.e., returns a $\gamma\eps$-approximate fixed point $x$)}{
        \KwRet $x$\tcp*[r]{This is an $\eps$-near fixed point by \cref{lem:triv bounds}}
    }
    \Else{
        \KwRet the returned $\tau$-accurate labeling\tcp*[r]{Label $-1$ is $\mathsf{lo}$, label $+1$ is $\mathsf{hi}$}
    }
\end{algorithm}

We devote the rest of the subsection to proving \cref{thm:amplify}.
We prove the theorem by induction on the recursion depth.
The induction is well-founded since $\rho>0$ and every recursive
call reduces $\tau$ by at least a factor of two.
Moreover, each boosting iteration multiplies $B_+B_-$ by $3/4$,
so the boosting loop terminates after finitely many iterations.
We analyze the two sides of the contract separately.
For each side, every boosting iteration either correctly assigns
all required labels, which then remain fixed, or preserves the
corresponding potential bound for the next iteration.
The base cases on \cref{line:base-oracle,line:base-decode} are immediate. Furthermore, if at any point any recursive call returns an admission of defeat from $\Oracle$, it is automatically propagated. Hence, we focus on the case in which all recursive calls return labelings.
We prove the guarantee for the $+1$ labels (i.e., $\Phi_+(F) \le B_+ \implies \ell_i = +1$ for all $i \in R_+$); the analysis for the $-1$ labels is symmetric. The main idea of the proof is captured by the following key lemma (a symmetric result holds for the $-1$ labels and is proved identically).
\begin{lemma}\label{lem:progress}
    Each iteration of the boosting loop (\cref{line:boost-loop} in \Cref{algo:amplify}) that starts with $\Phi_+(F) \le B_+$ either:
    \begin{enumerate}[label={\normalfont(\roman*)}]
        \item sets $+1$ labels to \emph{all} coordinates in $R_+$ via \cref{line:label-plus}; or
        \item guarantees that the sum of exponential weights of the misclassified $+1$ labels (i.e., labels in $S_+$ that were voted $-1$ instead of the correct $+1$ on \cref{line:weak-learner}) is less than $1/4$ of $\Phi_+(F)$. Thus, the inductive contract of boosting \eqref{eq:boosting relaxed} holds, and the updated bound $B_+ \gets \frac{\sqrt 3}{2} B_+$ remains a valid upper bound for $\Phi_+(F)$ at the next round. 
    \end{enumerate}
\end{lemma}
\begin{proof} %
    Consider a generic iteration of boosting  (\crefrange{line:boost-init}{line:weak-learner}), excluding the direct update of vote tallies. The values of $b, \tau,$ and $F$ remain constant, and so do the shifted values $\sigma'_i = \sigma_i+b_i$, $H'=\max_i \sigma'_i$, and exponential weight potential 
    \[
        \Phi_+(F) = \sum_{i: H' - \sigma' < 8\tau} \sqrt{3}^{-F_i} \le B_+.
    \]
    
    Throughout the iteration, the values of $\delta_i$ change. We denote the shifted slacks given the current shift vector $b+\delta$ supplied to the recursive calls on \cref{line:step1-call,line:step2-call,line:weak-learner} as
    \[
        \sigma''_i \coloneqq  \sigma_i + b_i + \delta_i,\qquad H'' \coloneqq \max_i \sigma''_i.
    \]
    Furthermore, we denote with
    \[
        g''_i \coloneqq H'' - \sigma''_i
    \]
    the gap of $\sigma''_i$ from the maximum. 

    \paragraph{A technical observation.} We start with a technical observation that will be necessary for invoking the inductive contract (\cref{contract:main}) in the calls on \cref{line:step1-call,line:step2-call,line:weak-learner}.
    \cref{contract:main} gives guarantees conditional on the exponential weight potential, which is a function of the supplied shifts and tolerance. In all recursive calls, the supplied shifts are $b + \delta$, while the tolerance is $\tau/32$ for \cref{line:step1-call,line:step2-call}, and $\tau/2$ for \cref{line:weak-learner}.

    To connect the inductive contract condition in the recursive call to that of the current call, we therefore have to relate the exponential weight potentials $\Phi^{\delta,\tau''}_+(F)$ for $b + \delta$ at tolerance $\tau''=\tau/32$ and $\tau''=\tau/2$. We claim that 
    \begin{equation}\label{eq:recursive potential}
        \Phi_+^{\delta,\tau'}(F) \le \Phi_+(F).
    \end{equation}
    This is quite convenient, as it means that a recursive call with unmodified bound $B_+$ (as is the case on \cref{line:step1-call,line:weak-learner}) automatically satisfies the precondition of the inductive contract.

    To prove \eqref{eq:recursive potential}, expand the definition
    \[
        \Phi_+^{\delta,\tau''}(F) = \sum_{i: H'' - \sigma''_i < 8 \tau''} \sqrt{3}^{-F_i}.
    \]
    The proof is complete if we can show that $\{i: H'' - \sigma''_i < 8\cdot \tau/2\} \subseteq \{i : H' - \sigma' < 8\tau\}$. By direct inspection of \cref{algo:amplify}, it is clear that every $\delta_i$ must be in the range $[-2\tau, 2\tau]$. Hence, $H' \le H'' + 2\tau, \sigma' \ge \sigma'' - 2\tau$ and thus $H' - \sigma' \le H'' - \sigma'' + 4\tau$, completing the proof.
    
    \paragraph{Step I.} During \texttt{Step I} (\crefrange{line:step1-loop}{line:shift-up}), the recursive call to $\Amplify$ uses the bound $B_+$.   
    Hence, from \eqref{eq:recursive potential} the precondition of the inductive contract (\cref{contract:main}) holds, and the recursive call labels every coordinate with
    \[
        \sigma''_i > H'' - \frac{\tau}{32}
    \]
    with label $+1$. Hence, none of those coordinates sees an increased $\delta$, while any other coordinate can rise by $\frac{\tau}{32}$ starting from a value no greater than $H'' - \frac{\tau}{32}$. This shows that, for all $i\in\{1,\dots,n\}$,
    \begin{equation}\label{eq:H1}
        \text{the value of $H''$ does not change during \texttt{Step I}, and $g''_i$ never increases.}
    \end{equation}

    \paragraph{Step II.} Consider now a generic iteration in \texttt{Step II}. $\Amplify$ is called recursively with a bound on $\Phi_+$ set to $\frac34 B_+$, and accuracy $\frac{\tau}{32}$. Hence, expanding the definition of the exponential potential from the point of view of that call, \cref{contract:main} guarantees that the $+1$ labels returned by the recursive call when
    \begin{equation}\label{eq:light}
        \sum_{i: H'' - \sigma''_i \le \frac\tau4} \sqrt{3}^{-F_i} \le \frac{3}{4}B_+.
    \end{equation}
    When that is the case, we call the iteration \emph{light}, because it means that the exponential weight at the extremal band $\{i: H'' - \sigma''_i \le \frac\tau4\}$ is at most $3/4$ of the bound, and thus the exponential weight is not too concentrated at the very ends of $\sigma''$. In contrast, we call the iteration \emph{heavy} when \eqref{eq:light} does not hold.

    If the iteration is light, the inductive contract guarantees that all the coordinates within $\frac\tau{32}$ from the maximum are correctly labeled $+1$, and thus they are pushed down by $\frac\tau{32}$. Thus, for all $i$,
    \begin{equation}\label{eq:H2L}
        \text{each light iteration in \texttt{Step II} decreases $H''$ by exactly $\frac\tau{32}$, and $g''_i$ never increase.}
    \end{equation}
    
    If the iteration is heavy, no guarantee is made about the labeling $\ell'$ returned by the recursive call to $\Amplify$ on \cref{line:step2-call}. However, since the $\delta_i$ are only decreased on \cref{line:shift-down},
   \begin{equation}\label{eq:H2H}
        \text{each heavy iteration in \texttt{Step II} cannot increase the maximum $H''$.}
    \end{equation}
    Furthermore, \cref{line:shift-down} only keeps the $\sigma''_i$ unchanged or decreases it by $\frac{\tau}{32}$, for all $i$
    \begin{equation}\label{eq:H2G}
        \text{each heavy iteration in \texttt{Step II} increases $g''_i$ by at most $\frac{\tau}{32}$.}
    \end{equation}

    \paragraph{Step III.} 
    As established in \eqref{eq:H1}, \eqref{eq:H2L}, and \eqref{eq:H2H}, \texttt{Steps I} and \texttt{II} never increase the maximum $H''$. Hence, for any coordinate $i$ that requires a $+1$ label (i.e., $i \in R_+$) at the beginning of \texttt{Step III},
    \[
        \delta_i = \sigma''_i - (\sigma_i + b_i) \le H'' - (\sigma_i + b_i) \le H' - (\sigma_i + b_i) = H' - \sigma'_i < \tau.
    \]
    Hence, no $i \in R_+$ can ever be mislabeled as $-1$ on \cref{line:label-minus}.

    If all 64 loops of \texttt{Step II} were light, the value of $H''$ on \cref{line:direct-labels} would be exactly $H' - 64 \frac{\tau}{32} = H' - 2\tau$. Hence, every $i \in R_+$ would satisfy
    \[
        \delta_i = \sigma''_i - (\sigma_i + b_i) \le H'' - (\sigma_i + b_i) = H' - \sigma'_i - 2\tau < \tau - 2\tau = -\tau
    \]
    and therefore \cref{line:label-plus} correctly labels all coordinates in $R_+$ as $+1$. This shows case (i) in the statement.

    \paragraph{Step IV.}
    Contrapositively, if by the end of \texttt{Step III} at least one label of the coordinates in $R_+$ remains unset ($\bot$), then at least one of the iterations in \texttt{Step II} was heavy. We will then show that the single recursive call on \cref{line:weak-learner} provides a weak learner whose exponential weight mass placed on misclassified labels $\ell'$ is at most $\frac14 B_+$.

    Consider the \emph{last} heavy iteration of \texttt{Step II}. By definition of a heavy iteration, \emph{at the beginning of that iteration}, the opposite of \eqref{eq:light} holds, that is,
    \[
        \sum_{i: H'' - \sigma''_i \le \frac\tau4} \sqrt{3}^{-F_i} > \frac{3}{4}B_+.
    \]
    Let $G$ denote the set $\{i: H'' - \sigma''_i \le \frac\tau4\}$ at the beginning of that heavy iteration.
    In light of \eqref{eq:H2G}, at the end of that iteration, the gap $g''_i = H'' - \sigma''_i$ has increased by at most $\frac{\tau}{32}$. All subsequent iterations are light, and by \eqref{eq:H2L} cannot increase the gap $g''_i$. Hence, just before the call to $\Amplify$ in \texttt{Step IV}, the gaps of all $i \in G$ satisfy
    \[
        g''_i \le \frac{\tau}{4} + \frac{\tau}{32} = \frac{9\tau}{32} < \frac12 \tau.
    \]
    This implies that all $i \in G$ are within the band $R_+$ of the recursive call on \cref{line:weak-learner}. 
    Since that call uses the unscaled budget $B_+$, and $\Phi_+(F) \le B_+$ by hypothesis, the inductive \cref{contract:main} together with \eqref{eq:recursive potential} guarantee that all $i \in G$ get correctly labeled as $+1$. The exponential weight of the labels that $\ell'$ puts as $-1$ is therefore
    \[
        E_+ \coloneqq \sum_{i \in S_+ : \ell'_i = -1} \sqrt{3}^{-F_i} \le \Phi_+(F) - \sum_{i \in G} \sqrt{3}^{-F_i} < \Phi_+(F) - \frac{3}{4}B_+ \le \frac{1}{4}B_+.
    \]
    This completes the proof of case (ii) in the statement.
\end{proof}

With \cref{lem:progress}, proving that \cref{algo:amplify} terminates with correct $+1$ label for $R_+$ when $\Amplify$ is called with $\Phi_+(F) \le B_+$ is straightforward. In the first boosting iteration, $\Phi_+(F) \le B_+$ and we can invoke \cref{lem:progress}. If case (i) materializes, the labels in $R_+$ are set correctly and will never be changed, so we are done. If case (ii) materializes, the recursive call on \cref{line:weak-learner} is a weak oracle. Thus, we can set $\Phi_+(F) \le B_+$ according to \eqref{eq:boosting relaxed}, and guarantee that the next boosting iteration starts with a valid bound $\Phi_+(F) \le B_+$ after the vote tallies are updated on \cref{line:vote-update}. By the end of the boosting round, \Cref{lem:decoding} guarantees that the unset labels of elements in $R_+$, if any, will be correctly set to $+1$.

Combining the above guarantee for $R_+$ with the symmetric guarantee for $R_-$ completes the proof of \cref{thm:amplify}.

\subsection{Complexity of $\Amplify$}
We start by bounding the number of arithmetic operations\footnote{To expose the important ideas without getting bogged down in bit representation considerations, we analyze the complexity assuming a unit-cost real machine. Bounding the bit complexity of the algorithm is straightforward and brings no surprises.\label{fn:bits}} and calls to $\Oracle$ performed by $\Amplify$.
Each recursive call to $\Amplify$ cuts the tolerance $\tau$ by a factor of $2$ or $32$ (\cref{line:weak-learner,line:step1-call,line:step2-call}), so the recursion depth is $O(\log(\tau/\rho))$. Furthermore, each call to $\Amplify$ in \texttt{Step I} and \texttt{Step II} cuts the product $B_+B_-$ by a factor of $3/4$, and so does the bound update on \cref{line:potential-update}.
By combining these observations, we can bound the total number of boosting rounds and recursive calls to $\Amplify$ in terms of the logarithms of the initial budget product $B_+B_-$ and the ratio $\tau/\rho$.

\begin{lemma}[Cost of amplification]\label{lem:amplify cost}
Fix $n\ge2$ and $\Oracle$ accuracy $\rho>0$. Consider the call $\Amplify(b,\tau,F,B_+,B_-)$ with $\tau\ge\rho$ and $B_+B_- \ge 1$, integer votes $F\in\mathbb Z^n$, and positive potential bounds $B_+,B_-$. 
The total number of boosting rounds across the entire computation is at most
\begin{equation}\label{eq:amplify call bound}
    N\coloneqq130^k\binom{h+k}{k}
    \le [130(h+1)]^k,
\end{equation}
where $k = 1 + \lceil\log_{4/3}(B_+B_-)\rceil$ and $h = \lceil\log_2(\tau/\rho)\rceil$.
The number of arithmetic operations is therefore at most $O(Nn)$, and the number of calls to $\Oracle$ is at most $129N$.
\end{lemma}

\begin{proof}
Let $P \coloneqq B_+B_-$ denote the product of the potential bounds. Each boosting round calls 128 times $\Amplify$ with a reduced product $\nicefrac34 P$ and tolerance $\tau/32$, and once with the unmodified $P$ and tolerance $\tau/2$. Then, each bound $B_+, B_-$ is multiplied by $\sqrt{3}/2$, so the product $P$ decreases by a factor $3/4$ and the loop continues. Denoting with $A(P, \tau)$ the total number of boosting rounds (including in recursive calls) starting from product $P$ and tolerance $\tau$, we can write the recurrence
\begin{align*}
    A(P, \tau) &\le 1 + A(P, \tau/2) + 128\cdot  A(\nicefrac34 P, {\tau}/{32}) + A(\nicefrac34 P, \tau)  \\
     &\le 1 + A(P, \tau/2) + 129\cdot  A(\nicefrac34 P, \tau),
\end{align*}
where in the second line we used the monotonicity of $A(P, \tau)$ in $\tau$. The base cases are $A(P, \tau) = 0$ for $\tau \le \rho$ or $P < 1$. The claimed bound \eqref{eq:amplify call bound} can then be verified by induction using Pascal's identity $\binom{h+k}{k} = \binom{h-1+k}{k} + \binom{h-1+k}{k-1}$.
\end{proof}

\cref{thm:amplify cost} then follows as a direct corollary.

\begin{proof}[Proof of \cref{thm:amplify cost}]
Correctness is immediate in light of \cref{contract:main} and \cref{thm:amplify}, applied to the initial call $\Amplify(0,2\kappa,0,n,n)$.
By \cref{lem:amplify cost}, the number of boosting iterations (and thus, up to a factor 129 also the number of calls to $\Oracle$) satisfies
\[
    Q \le (C_1 (1 + \log(2\kappa/\rho)))^{C_2 \log n}
\]
where $C_1$ and $C_2$ are universal constants. Using the hypothesis that $\kappa / \rho \ge 4$, we can therefore absorb the addition and multiplication by $C_1$ into a larger constant in the exponent, yielding the statement.
\end{proof}

\subsection{Complexity of $\Solve$}\label{sec:main proof}
Finally, we finish bounding the complexity of $\Solve$. As a reminder, $\Solve$ calls $\Label$ a total $\log_{8/7} (1/\eps)$ times. Each call to $\Label$ consists of a call to $\Amplify$ for the oracle provided by $\ThinLabel$. $\ThinLabel$ in turn calls $\Solve$ recursively on at most two problems of half the dimension.  

\begin{repeatthm}{thm:main}
Fix $d\ge2$, a contraction gap $\gamma\in(0,1]$, and a desired nearness tolerance $0<\eps\le1/4$. For a $(1-\gamma)$-contraction $f:[0,1]^d\to[0,1]^d$, the total number of arithmetic operations and exact evaluations of $f$ performed by $\Solve(f,\eps)$ to return an $\eps$-\emph{near} fixed point of $f$ is upper bounded by
\[
\left(\log\frac1{\gamma\eps}\right)^{C \log^2 d},
\]
for a universal constant $C$.
\end{repeatthm}

\begin{proof}
Throughout the proof, we let $C$ denote a universal constant that may change from line to line.

Each call to $\Solve$ with a $d$-dimensional map $f$ and tolerance $\eps$ performs at most $C \log_2(1/\eps)$ calls to $\Label$ with the same $f$ and tolerance at most constant, and a number $Cd \log_2(1/\eps)$ of arithmetic operations. Each call to $\Label$ makes a single call to $\Amplify$ (\cref{algo:label}), with amplification ratio $\kappa/\rho \le 2/(\gamma^2\eps)$. $\Amplify$ calls $\Oracle$ at most one plus $129$ times the number of boosting rounds, which is bounded by \cref{lem:amplify cost}; the number of arithmetic operations in each boosting round is linear in the dimension. Finally, each call to $\Oracle$ is provided by $\ThinLabel$ with accuracy $\eta=\gamma\eps$, which in turn makes at most two recursive calls to $\Solve$ on problems of dimension at most $\lfloor d/2\rfloor$ and tolerance $\gamma\eta/2 = \gamma^2\eps/2$, as well as at most a quadratic number of additional operations to determine the median of the box dimensions (\cref{algo:thin}). Once the dimension of the problem drops to $d=1$, the $\ell_\infty$-contraction is solved by bisection, which takes $O(\log(1/\eps))$ arithmetic operations and evaluations of $f$. In dimension zero, $\Solve$ returns the empty vector.

Each evaluation of a restricted map also takes at most $Cd$ additional operations to form the input and restrict the output. Combining these observations, and taking the worst case over maps of dimension at most $d$, we can write a recurrence for the total number of arithmetic operations, comparisons, and exact evaluations of $f$ performed by $\Solve(f,\eps)$:
\[
\begin{aligned}
    T_d(\eps)&\le Cd^2\log_2(1/\eps)\,U_d({2}/{\gamma^2\eps})
    (d^2+T_{\lfloor d/2\rfloor}(\gamma^2\eps/2)),\\
    T_1(\eps)&\le C\log_2(1/\eps).
\end{aligned}
\]
where $U_d(a)$ bounds one plus the number of boosting rounds.

Let
\[
    r\coloneqq\lceil\log_2d\rceil,\qquad
    \ell\coloneqq2+\log_2\frac1{\gamma\eps},\qquad
    \eps_j\coloneqq\eps\left(\frac{\gamma^2}{2}\right)^j.
\]
After $j$ steps of the recurrence, the dimension is at most $2^{r-j}$ and the tolerance is $\eps_j$. For $j<r$, write $t=r-j$. The parameters in \cref{lem:amplify cost} satisfy $h\le Cr\ell$ and $k\le Ct$. Increasing $k$ to an integer $Ct$ in \eqref{eq:amplify call bound} and using $\binom{h+k}{k}\le[C(1+h/k)]^k$, we obtain
\[
    U_{2^t}\left(\frac{2}{\gamma^2\eps_j}\right)
    \le\left(\frac{Cr\ell}{t}\right)^{Ct}.
\]
Since $\log_2(1/\eps_j)\le Cr\ell$ for $j\le r$, iterating the recurrence gives
\begin{equation}\label{eq:solve-work}
\begin{aligned}
    T_d(\eps)
    &\le C (d^2r\ell)^r r\ell
        \prod_{t=1}^r\left(\frac{Cr\ell}{t}\right)^{Ct}\\
    &\le C 2^{r^2}\ell^{Cr^2}\le\ell^{Cr^2} \le \log(1/(\gamma\eps))^{C\log^2 d},
\end{aligned}
\end{equation}
where we used the fact that $\gamma\eps\le1/4$ to absorb terms into a large exponent constant $C$.
\end{proof}

\section*{Acknowledgments}

The paper was written entirely by humans and we take full responsibility for its content. AI served as an important amplifier as we explored different ideas and designs over several months. The final result is a synthesis between human and artificial ideas.

This work was supported in part by the National Science Foundation awards CCF-2443068 and IIS-2552046, the Office of Naval Research grant N000142512296, and an AI2050 Early Career Fellowship.

\bibliographystyle{alpha}
\bibliography{refs.bib}

\end{document}